\documentclass[pra,aps,superscriptaddress,twocolumn,notitlepage,letter,nopacs,nofootinbib,longbibliography]{revtex4-2}

\usepackage{graphicx,color,amsmath,amsfonts,enumerate,amsthm,amssymb,mathtools,enumitem,thmtools,hyperref,subfigure,mathdots,enumitem,centernot,bm,soul,bbm,stackrel}
\usepackage[capitalise, noabbrev]{cleveref}
\usepackage{mathrsfs}
\usepackage{babel}
\usepackage{stmaryrd}
\usepackage{multirow}
\hypersetup{colorlinks=true,linkcolor=blue,citecolor=blue,urlcolor=blue}

\usepackage{fontawesome}

\usepackage{tikz,pgfplots}
\usetikzlibrary{quantikz}

\def\B{ {\cal B} }
\def\C{ {\mathscr C} }
\def\D{ {\cal D} }
\def\E{ {\cal E} }
\def\F{ {\cal F} } 

\def\H{ {\cal H} }
\def\I{ {\cal I} }

\def\M{ {\cal M} }
\def\N{ {\cal N} }

\def\S{ {\cal S} }
\def\T{ {\cal T} }
\def\U{ {\cal U} }
\def\V{ {\cal V} }

\def\>{\rangle}
\def\<{\langle}
 
\newcommand{\ketbra}[2]{\ensuremath{\left|#1\right\rangle\!\!\left\langle#2\right|}}
\newcommand{\matrixel}[3]{\ensuremath{\left\langle #1 \vphantom{#2#3} \right| #2 \left| #3 \vphantom{#1#2} \right\rangle}}
\newcommand{\tr}[1]{\mathrm{Tr}\left( #1 \right)}
\newcommand{\trr}[2]{\mathrm{Tr}_{#1}\left( #2 \right)}
\newcommand{\iden}{\mathbbm{1}}

\newcommand{\kkhide}[1]{}

\renewcommand{\v}[1]{\ensuremath{\boldsymbol #1}}

\definecolor{ppblue}{RGB}{46,117,182}
\definecolor{ppred}{RGB}{197, 90, 17}

\theoremstyle{plain}
\newtheorem{thm}{Theorem}
\newtheorem{lem}[thm]{Lemma}
\newtheorem{prop}{Proposition}

\theoremstyle{definition}
\newtheorem{defn}{Definition}

\begin{document}

\title{Dephasing superchannels}

\author{Zbigniew Pucha{\l}a}
\affiliation{Institute of Theoretical and Applied Informatics, Polish Academy of Sciences, 44-100 Gliwice, Poland}
\affiliation{Faculty of Physics, Astronomy and Applied Computer Science, Jagiellonian University, 30-348 Krak\'{o}w, Poland}

\author{Kamil Korzekwa}
\affiliation{Faculty of Physics, Astronomy and Applied Computer Science, Jagiellonian University, 30-348 Krak\'{o}w, Poland}

\author{Roberto Salazar}
\affiliation{Faculty of Physics, Astronomy and Applied Computer Science, Jagiellonian University, 30-348 Krak\'{o}w, Poland}

\author{Pawe{\l} Horodecki}
\affiliation{International Centre for Theory of Quantum Technologies, University of Gda\'{n}sk, Wita Stwosza 63, 80-308 Gda\'{n}sk, Poland}

\author{Karol {\.Z}yczkowski}
\affiliation{Faculty of Physics, Astronomy and Applied Computer Science, Jagiellonian University, 30-348 Krak\'{o}w, Poland}
\affiliation{Center for Theoretical Physics, Polish Academy of Sciences, 02-668 Warszawa, Poland}

\begin{abstract}
	We characterise a class of environmental noises that decrease coherent properties of quantum channels by introducing and analysing the properties of dephasing superchannels. These are defined as superchannels that affect only non-classical properties of a quantum channel $\E$, i.e., they leave invariant the transition probabilities induced by $\E$ in the distinguished basis. We prove that such superchannels $\Xi_C$ form a particular subclass of Schur-product supermaps that act on the Jamio\l{}kowski state $J(\E)$ of a channel $\E$ via a Schur product, $J'=J\circ C$. We also find physical realizations of general $\Xi_C$ through a pre- and post-processing employing dephasing channels with memory, and show that memory plays a non-trivial role for quantum systems of dimension $d>2$. Moreover, we prove that coherence generating power of a general quantum channel is a monotone under dephasing superchannels. Finally, we analyse the effect dephasing noise can have on a quantum channel $\E$ by investigating the number of distinguishable channels that $\E$ can be mapped to by a family of dephasing superchannels. More precisely, we upper bound this number in terms of hypothesis testing channel divergence between $\E$ and its fully dephased version, and also relate it to the robustness of coherence of $\E$.
\end{abstract}
\date{August 23, 2021}
\maketitle

%-------------------------------------------------------------------------------
% SEC. I - INTRODUCTION
%-------------------------------------------------------------------------------

\section{Introduction}

Quantum technologies bring the promise of revolutionising the way we process information by employing quantum effects, such as superposition and entanglement to overcome current limitations of information processors~\cite{QT1,QT2}. However, these quantum effects are extremely fragile to noise and any potential quantum advantage disappears in the presence of uncontrolled interactions with the environment~\cite{QEC0}. Thus, the biggest obstacle on the way to constructing practical quantum devices is to harness noise and decoherence effects. Although a lot depends on the development of experimental techniques to control quantum systems, theoretical investigations can also bring progress in that field. One of the main approaches to achieve it is to develop novel quantum error-correcting codes that allow one to protect quantum information against the deteriorating effects of noise~\cite{QEC0,QEC1,QEC2,QEC3}. A complementary path, which we will follow in this paper, is to study the mathematical structure of significant noise models in order to better understand their properties and the way they affect quantum information.

\begin{figure*}[t]
    \centering
    \includegraphics[width=0.98\textwidth]{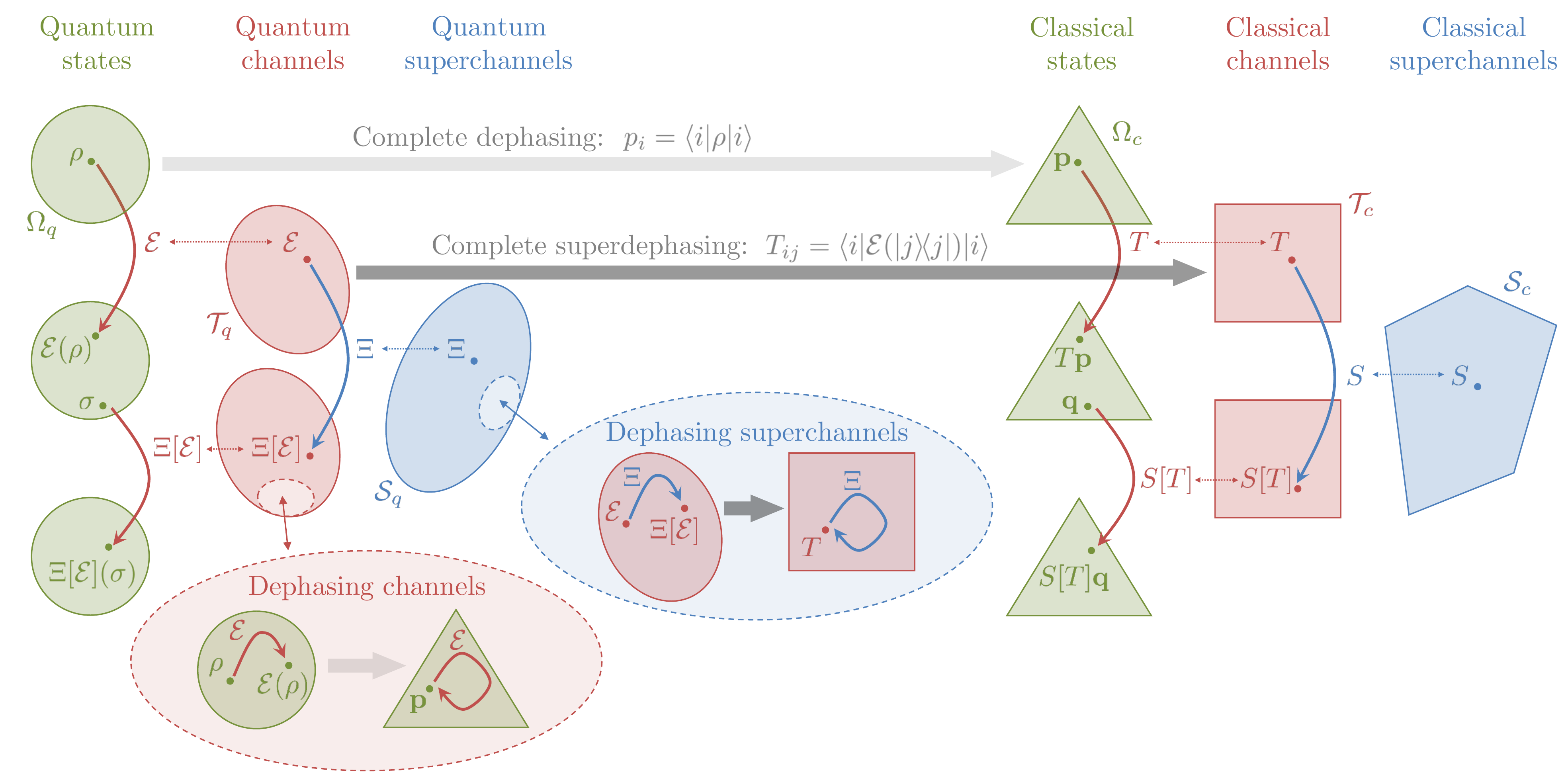}
    \caption{\label{fig:overview}\textbf{Dephasing channels and superchannels.} The set of quantum states $\Omega_q$ (density matrices) is projected onto the set of classical states $\Omega_c$ (probability distributions) via the completely dephasing channel that removes all coherences in the distinguished basis, but keeps the occupations unchanged. General dephasing channels are those maps between quantum states that affect coherences, but do not change occupations, i.e., they keep the classical (completely dephased) version of the state invariant. Analogously, the set of quantum channels $\T_q$ (completely positive trace-preserving maps) is projected onto the set of classical channels $\T_c$ (stochastic matrices) via a completely dephasing superchannel that removes all coherent properties of the channel, but keeps the transition probabilities in the distinguished basis unchanged. General dephasing superchannels form the subset of quantum superchannels $\S_q$ that affect coherent properties of the channel, but do not change transition probabilities, i.e., they keep the classical (completely superdephased) version of the channel invariant. } 
\end{figure*}

An essential class of noises is given by dephasing processes~\cite{DP1,DP2}, i.e., processes that deteriorate the coherence of a quantum system in a distinguished basis, but do not affect occupations. They can be interpreted as a purely quantum noise because classical information processing is unaffected by dephasings. A first rigorous investigation of the capacity of such noise channels, under the name of generalised dephasing channels, was performed by Devetak and Shor \cite{devetak2005capacity}. Further studies along these lines have been performed \cite{d2007quantum,bradler2010trade}, since due to the structural simplicity of dephasing channels, single-letter formulas for their classical and quantum capacities could be found. Additionally, the practical relevance of a special class of dephasing channels was demonstrated in the context of quantum privacy~\cite{levick2017quantum}.

In all these previous works, the focus was on the effect the dephasing noise has on the state of the system. Here, we investigate the effect it has on quantum gates, i.e., we do not ask, how the state of the system gets affected, but how the whole dynamics changes in the presence of a dephasing noise. This forms an extension of previous works, since the effect noise has on a quantum gate $\E$ cannot be simply captured by pre- and post-processing by some noise channels $\N_1$ and $\N_2$:
\begin{equation}
	\label{eq:superchannel_simple}
	\begin{quantikz}
		& \gate{\tilde{\E}} &\qw
	\end{quantikz}
		=
	\begin{quantikz}
		\qw&\gate{\N_1}& \gate{\E} &\gate{\N_2}&\qw
	\end{quantikz}.
\end{equation}
This is due to potential correlations between the input and output states of the investigated gate $\E$ mediated by the environment, and so the general noisy version $\tilde{\E}$ of the gate $\E$ has the following form
\begin{equation}
	\label{eq:superchannel_general}
	\begin{quantikz}
		& \gate{\tilde{\E}} &\qw
	\end{quantikz}
		=
	\begin{quantikz}
		&[-0.1cm]	\gate[wires=2]{\N_1}&[-0.1cm] \gate{\E} &[-0.1cm] \gate[wires=2]{\N_2}&[-0.6cm]\qw \\[-0.5cm]
		\ketbra{0}{0}	&  	 & \qw 	&	& [-0.6cm]	\trash{\text{discard}}
	\end{quantikz}\!\!\!.
\end{equation}
A mathematical concept that can capture such a general effect of noise on quantum gates is a quantum superchannel~\cite{chiribella2008transforming}, also called a supermap and used to describe dynamics in generalized quantum theories~\cite{Zy08}. In this article, we introduce the notion of \emph{dephasing superchannels} as an analogue of dephasing channels: such superchannels should not affect the classical properties of the channel $\E$ they act upon, i.e., the transition probabilities induced by $\E$ in the distinguished basis should be invariant. We illustrate these concepts in Fig.~\ref{fig:overview}.

We first describe the mathematical structure of dephasing superchannels by relating them to a particular subset of Schur-product maps on Jamio{\l}kowski states. We also provide a physical realisation of such superchannels in the form of pre- and post-processing employing dephasing channels with memory, i.e., in the form of Eq.~\eqref{eq:superchannel_general} with $\N_1$ and $\N_2$ being directly related to dephasing channels. Moreover, we explicitly demonstrate that for system's dimension $d\geq 3$ this memory effect extends the set of possible dephasing noises. After describing these basic properties of dephasing superchannels, we then focus on the effect they have on coherent properties of quantum channels. We start by proving that the cohering power of a quantum channel always decreases under dephasing superchannels. We then proceed to analyze, how strongly a quantum channel can be perturbed by dephasing superchannels. More precisely, we provide an upper-bound for the number of distinguishable (orthogonal) channels that a given channel $\E$ can be steered to by dephasing noises, where the bound is given by a particular coherence measure of a channel $\E$. Finally, we give a complementary perspective on that problem, where coherence of a channel $\E$ can be seen as a resource for distinguishing between various dephasing superchannels.

The paper is organised as follows. In Sec.~\ref{sec:setting}, we recall the basic properties of quantum states, channels and superchannels. We also revisit the concept of a dephasing channel and relate it to the notion of a Schur-product superoperator. Then, in Sec.~\ref{sec:structure}, we introduce dephasing superchannels, present their mathematical structure and discuss physical realisations. The following Sec.~\ref{sec:coherence} contains the analysis of the interplay between the action of dephasing superchannels and the coherent properties of quantum channels. Finally, Sec.~\ref{sec:outlook} contains conclusions and outlook for future work.

%-------------------------------------------------------------------------------
% SEC. II - SETTING THE SCENE
%-------------------------------------------------------------------------------

\section{Setting the scene}
\label{sec:setting}

%-------------------------------------------------------------------------------
% SEC. II.A
%-------------------------------------------------------------------------------

\subsection{Quantum states, channels and superchannels}

A state of a $d$-dimensional quantum system is represented by a density operator $\rho$ acting on a $d$-dimensional Hilbert space $\H_d$. The set of density operators $\Omega_q$ forms a subset of bounded operators $\B(\H_d)$ that are positive semi-definite, $\rho\geq 0$, and have unit trace, $\tr{\rho}=1$. General linear transformations \mbox{$\B(\H_d)\rightarrow\B(\H_d)$} are called \emph{superoperators}, while their subset $\T_q$ corresponding to physical evolutions of quantum states is known as \emph{quantum channels}. These model all quantum gates and form a subset of superoperators that are completely positive (CP) and trace-preserving (TP). The evolution of a closed quantum system is described by a unitary channel, $\U(\cdot)=U(\cdot)U^\dagger$, with a unitary matrix $U$ of size~$d$.

A natural representation of a quantum channel $\E$ is given by a $d^2\times d^2$ matrix $\Phi(\E)$:
\begin{equation}
    \Phi(\E)_{ij,kl}:= \tr{\ketbra{i}{j}\E(\ketbra{k}{l})}
\end{equation}
with $\{\ket{i}\}_{i=1}^d$ being some fixed basis of $\H_d$. However, in this paper we will mostly use the following three alternative representations~\cite{nielsen2010quantum}, each useful for a different reason. First, through Stinespring dilation, every quantum channel can be realised by a unitary dynamics of an extended system followed by discarding the ancillary system:
\begin{equation}
    \mathcal{E}(\cdot)=\textrm{Tr}_{2}(U((\cdot)\otimes\ketbra{0}{0})U^{\dagger}).
    \label{eq:Stin}
\end{equation}
This gives a clear physical interpretation of the action of~$\E$. Second, one can employ the operator-sum representation of $\E$, which is particularly useful for performing calculations, and write
\begin{equation}
    \E(\cdot)=\sum_{i=1}^{d}K_{i}\rho K_{i}^{\dagger}\label{eq:kraus0},\quad \sum_{i=1}^{d}K_{i}^{\dagger}K_{i}=\iden,
\end{equation}
with $\{ K_{i}\}_{i=1}^{d}$ known as the \emph{Kraus operators} and $\iden$ denoting the identity matrix of size $d$. Finally, through Choi-Jamio{\l}kowski isomorphism~\cite{jamiolkowski1972linear,choi1975completely} one can represent a channel $\E$ via its Jamio{\l}kowski matrix $J(\E)$:
\begin{equation}
    J(\E):=\E\otimes \I (\ketbra{\Psi}{\Psi}),\quad \ket{\Psi}:=\frac{1}{\sqrt{d}}\sum_{i=1}^d \ket{ii},
\end{equation}
with $\I$ denoting the identity channel on the ancillary system of dimension $d$. Crucially, the complete positivity of $\E$ is equivalent to $J(\E)\geq 0$, while the trace-preserving condition gets mapped to $\trr{1}{J(\E)}=\iden/d$. Moreover, the Jamio{\l}kowski state $J(\E)$ (also called dynamical matrix or Choi matrix when unnormalised),
can be related to the matrix representation $\Phi(\E)$ via a reshuffling operation, which reorders the entries of the matrix,
\begin{equation}
    \Phi(\E)_{ij,kl}=dJ(\E)_{ik,jl}.\label{Phi2}
\end{equation}

In order to investigate the effect of dephasing noise on quantum gates, we will need appropriate maps describing transformations of quantum channels into quantum channels. General linear
maps between superoperators,
\begin{equation}
    \Xi:[\B(\H_{d})\rightarrow\B(\H_{d})]\rightarrow[\B(\H_{d})\rightarrow\B(\H_{d})]
\end{equation}
will be called \emph{supermaps}, while their subset $\S_q$ corresponding to maps between quantum channels is known as \emph{superchannels}~\cite{chiribella2008transforming}. A general superchannel has a standard physical realization in terms of pre- and post-processing with a memory system as in Eq.~\eqref{eq:superchannel_general}. As with quantum channels, there are several useful representations of quantum superchannels~\cite{gour2019comparison}, but in our work we will only employ the analogue of the Choi-Jamio{\l}kowski representation. Denoting the superoperator basis elements on which channels can be spanned by
\begin{equation}
	\E_{(ij),(kl)}(\cdot) = \bra{k}\cdot\ket{l} \  \ket{i}\!\bra{j},
\end{equation}
the Jamio\l{}kowski matrix of a general superchannel $\Xi$ is given by
\begin{align}
	{\bf J}_{\Xi} &= \sum_{ijkl} J(\mathcal{E}_{(ij), (kl)}) \otimes 
	J(\Xi[\mathcal{E}_{(ij), (kl)}]). \label{eq:super_jamiolkowski}
\end{align}
Moreover, the action of a superchannel $\Xi$ can be expressed through its Jamio\l{}kowski matrix as
\begin{equation} \label{eqn:super-action}
	J(\Xi[\E]) = d^2\trr{2}{  {\bf J}_{\Xi}   (\iden \otimes J(\E)^\top)}.
\end{equation}

%-------------------------------------------------------------------------------
% SEC. II.B
%-------------------------------------------------------------------------------

\subsection{Dephasing channels}
\label{sec:schur_channels}

In order to investigate dephasing superchannels, we first need to recall the notion of dephasing channels and describe their known properties.

\begin{defn}[Dephasing channel]
	A quantum channel $\D$ is called a dephasing channel if the occupations in the distinguished basis are invariant under $\D$:
	\begin{equation}
		\forall~\rho,\ket{i}:\quad \bra{i}\D(\rho)\ket{i}=\bra{i}\rho\ket{i}.
	\end{equation}	
\end{defn}

The above definition has a clear physical interpretation. However, in order to study dephasing channels and superchannels, it is convenient to introduce the central mathematical concept of Schur product
(also called Hadamard product or entry-wise product) between operators for a fixed distinguished basis.

\begin{defn}[Schur-product superoperator]
	A superoperator $\D_C$ is called Schur-product if for all $X\in\B_d$ we have
	\begin{equation}
		\D_C(X) = \sum_{i,j=1}^d X_{ij} C_{ij} \ketbra{i}{j}=:X\circ C,
	\end{equation}
	where $\{\ket{i}\}_{i=1}^d$ is the distinguished basis, $C$ is a matrix of size $d$, and $\circ$ denotes the Schur product in the distinguished basis.
\end{defn}

We now have the following known result~\cite{kye1995positive,li1997special} that specifies the properties of $C$ for $\D_C$ to be a quantum channel, and relates Schur-product superoperators with dephasing channels.

\begin{lem}[Schur-product channels]
	\label{lem:schur_chan}
	A Schur-product superoperator $\D_C$ is a quantum channel if and only if $C$ is a correlation matrix (positive matrix with $C_{ii}=1$ for all~$i$). Moreover, $\D$ is a dephasing channel if and only if it is a Schur-product channel.
\end{lem}
\begin{proof}
	Direct calculation shows that
	\begin{equation}
		\label{eq:choi_schur}
		J(\D_C)=\frac{1}{d}\sum_{ij} C_{ij}\ketbra{ii}{jj}.
	\end{equation}
	Thus, the positivity of $J(\D_C)$ is equivalent to the positivity of $C$, and the trace preserving condition, \mbox{$\trr{1}{J(\D_C)}=\iden/d$}, is equivalent to $C_{ii}=1$. Therefore, $C$ is a correlation matrix, and $\D_C$ clearly preserves the diagonal. Now, assume that some channel $\D$ preserves the diagonal. We then have
	\begin{equation}
		J(\D)_{ij,ij}=\frac{1}{d}\bra{i}\D(\ketbra{j}{j})\ket{i}=\frac{\delta_{ij}}{d}.
	\end{equation}
	From the above and the positivity of $J(\D)$ we conclude that $J(\D)$ has the form from Eq.~\eqref{eq:choi_schur}, and thus is a Schur-product channel.
\end{proof}

It is also known how to physically realise a general Schur-product channel.

\begin{lem}[Physical realisation of $\D_C$]
	Every Schur-product channel can be written as a unitary processing with an ancillary system of dimension $d$ as follows:
	\begin{equation}
		\label{eq:diagram_channel}
		\begin{quantikz}
			& \gate{\D_C} &\qw
		\end{quantikz}
		=
		\begin{quantikz}
			&[-0.1cm]	\gate[wires=2]{\U}&[-0.1cm] \qw \\[-0.5cm]
			\ketbra{0}{0}	&&  	[-0.6cm]	\trash{\text{\emph{discard}}}
		\end{quantikz},
	\end{equation}
	where
	\begin{equation}
		\label{eq:U_channel}
		\U(\cdot)=U(\cdot)U^\dagger,\quad U=\sum_{i=1}^d \ketbra{i}{i}\otimes U_i, 
	\end{equation}
	with $\{U_i\}$ being arbitrary unitaries of size $d$. The relation between $C$ and these unitaries is given by
	\begin{equation}
		\label{eq:C_chanel}
		C_{ij} = \bra{0} U_j^\dagger  U_i \ket{0}.
	\end{equation}
\end{lem}
\begin{proof}
	First, we assume the $\D_C$ has the form from Eq.~\eqref{eq:diagram_channel} and calculate
	\begin{align}
		\D_C(\rho) &= \sum_{ij}\trr{2}{\ketbra{i}{i}\rho\ketbra{j}{j}\otimes U_i\ketbra{0}{0}U_j^\dagger}
		\nonumber\\
		&= \sum_{ij} \rho_{ij} C_{ij} \ketbra{i}{j}=\rho\circ C.
	\end{align}
	Next, assume $\D_C$ is a Schur-product channel. Since $C$ is a correlation matrix it can be written as a Gram matrix:
	\begin{equation}
		C_{ij}=\braket{\psi_j}{\psi_i}=:\bra{0}U_j^\dagger U_i\ket{0}
	\end{equation}
	for some choice of unitary matrices $\{U_i\}$.
\end{proof}

As a direct corollary of the above lemma, we can also obtain the following Kraus representation of a general Schur-product channel:
\begin{equation}
	\label{eq:kraus}
	\D_C(\cdot)=\sum_{k=1}^d K_k (\cdot)K_k^\dagger, \quad K_k = \sum_{i=1}^d  \braket{k}{\psi_i} \ketbra{i}{i},
\end{equation}
where $\ket{\psi_i}:=U_i\ket{0}$. 

Now, to see even more clearly the physical relevance of Schur-product channels and to understand why they are called dephasing channels, let us relate their action to the von Neumann measurement scheme. Let $\{\ket{i}\}$ denote the eigenstates of the measured observable and $\ket{\phi}$ the state of the system before the measurement. The von Neumann measurement scheme then involves a measuring apparatus in the initial state $\ket{0}$, its unitary interaction $U$ with the system,
\begin{equation}
	\ket{\phi}\otimes \ket{0} = \sum_i \braket{i}{\phi} \ket{i}\otimes \ket{0} \xrightarrow{~U~}  \sum_i \braket{i}{\phi} ~\ket{i}\otimes \ket{\psi_i},
\end{equation}
and the final projective measurement of the apparatus. From the above it is clear that $U$ has exactly the same form as the unitary in Eq.~\eqref{eq:U_channel}, and so the action of the von Neumann measurement on the measured system (after discarding the result) is given by a Schur-product channel $\D_C$. The correlation matrix $C$ describes on the one hand how much information about the system is encoded in the apparatus; and on the other, how much it disturbs (dephases) the system.

Finally, let us note two important properties of Schur-product channels in relation to resources of coherence and entanglement. First, by noting that each Kraus operator of $\D_C$ maps an incoherent state into an (unnormalized) incoherent state, we conclude that Schur-product channels belong to the set of incoherent operations~\cite{baumgratz2014quantifying}. Actually, they also belong to smaller subsets of incoherent operations, e.g. strictly incoherent operations and phase-covariant operations~\cite{streltsov2017colloquium}. As a result, all meaningful coherence  measures  cannot increase under the action of Schur-product channels which is yet another way to justify denoting them as dephasing channels. Second, by direct calculation, one can show that the channel complementary to $\D_C$ is a measure and prepare (and so an entanglement-breaking) channel:
\begin{align}
	\label{eq:complementary}
	\D^c_C(\rho)&:=\trr{1}{U (\rho\otimes \ketbra{0}{0})U^\dagger}=\sum_{i=1}^d \rho_{ii} \ketbra{\psi_i}{\psi_i}.
\end{align}

%-------------------------------------------------------------------------------
% SEC. III - STRUCTURE AND PROPERTIES OF DEPHASING SUPERCHANNELS
%-------------------------------------------------------------------------------

\section{Structure and properties of dephasing superchannels}
\label{sec:structure}

%-------------------------------------------------------------------------------
% SEC. III.A
%-------------------------------------------------------------------------------

The central object investigated in this paper is defined as follows.

\begin{defn}[Dephasing superchannel]
    A quantum superchannel $\Xi$ is called a dephasing superchannel if the transition probabilities in the distinguished basis are invariant under $\Xi$:
	\begin{equation}
    	\forall~\E,\ket{i},\ket{j}:\quad \bra{i}\Xi[\E](\ketbra{j}{j})\ket{i}=\bra{i}\E(\ketbra{j}{j})\ket{i}.
	\end{equation}	
\end{defn}

In what follows we first identify the above class of superchannels with a particular family of Schur-product supermaps and present physical realisation of every such superchannel. We also explain how noises generated by dephasing superchannels are more general than the ones generated by pre- and post-processing with dephasing channels. We finish this section by describing the particularly simple effect that dephasing superchannels have on dephasing channels.

\subsection{Equivalence with Schur-product superchannels}

In analogy to Schur-product superoperators, one can introduce the concept of Schur-product supermaps.

\begin{defn}[Schur-product supermaps]
	A supermap $\Xi_C$ is called Schur-product if for all \mbox{$X\in[\B(\H_d)\rightarrow\B(\H_d)]$} we have
	\begin{equation}
	\!\!	J(\Xi_C[X]) = \sum_{ijkl} J(X)_{ij,kl} C_{ij,kl} \ketbra{ij}{kl}= J(X)\circ C,
	\end{equation}
	where $J(X)$ is the Jamio\l{}kowski operator of $X$, $\{\ket{ij}\}$ is the distinguished basis, $C$ is a matrix of size $d^2$, and $\circ$ denotes Schur product in the distinguished basis.
\end{defn}
The above definition does not guarantee that the output of $\Xi_C$ will be completely positive and trace-preserving. Thus, a constraint on $C$ is given in the following proposition that also establishes equivalence between dephasing superchannels and Schur-product superchannels.
	
\begin{prop}[Schur-product superchannels]
	\label{prop:schur_super}
	A Schur-product supermap $\Xi_C$ is a quantum superchannel if and only if $C$ is a correlation matrix (positive matrix with all diagonal entries equal to 1) of the following form
	\begin{equation}
	\label{eq:corr_form}
	C=\left[\begin{array}{cccc}
	{C_{11}} & {C_{12}} & {\ldots} & {C_{1 N}} \\
	{C_{21}} & {C_{11}} & {\ldots} & {C_{2 N}} \\
	{\vdots} & {\vdots} & {\ddots} & {\vdots} \\
	{C_{N 1}} & {C_{N 2}} & {\ldots} & {C_{11}}
	\end{array}\right],
	\end{equation} 
	where $C_{ij}$ are $d\times d$ matrices and $C_{11}$ is a correlation matrix. Moreover, $\Xi$ is a dephasing superchannel if and only if it is a Schur-product superchannel.
\end{prop}	

\begin{proof}

We start by noting that for any correlation matrix~$C$ (positive semi-definite with ones on diagonal) of the form from Eq.~\eqref{eq:corr_form}, and any Jamio{\l}kowski operator $J(\E)$ of a CPTP map $\E$, \mbox{$J(\E) \circ C$} is also a Jamio{\l}kowski operator of a CPTP map. The above statement can be verified by direct inspection, since the Schur-product of positive semi-definite matrices is also positive semi-definite, and the partial trace over the first subsystem yields 
\begin{equation}
\trr{1}{J(\E) \circ C} = \trr{1}{J(\E)} \circ C_{11} = \iden/d.
\end{equation}

Now, we will assume that a Schur-product supermap $\Xi_C$ is a quantum superchannel and we will show that the matrix $C$ defining it has the form from Eq.~\eqref{eq:corr_form}. First, we employ Eq.~\eqref{eq:super_jamiolkowski} to write the Jamio{\l}kowski operator of $\Xi_C$ as
	\begin{align}\label{eqn:super=jamiolkowski}
		{\bf J}_{\Xi_C} &=\frac{1}{d} \sum_{ijkl} \ketbra{ik}{jl}\otimes 
		J(\mathcal{E}_{(ij), (kl)})\circ C \nonumber \\
		&=
		\frac{1}{d^2}\sum_{ijkl} C_{ik,jl} \ketbra{ik}{jl} \otimes 
		 \ketbra{ik}{jl}.
	\end{align}
Note that from our assumption and results of Ref.~\cite{chiribella2008transforming,gour2019comparison}, the matrix ${\bf J}_{\Xi_C} $ is positive semi-definite and therefore $C$ must be positive semi-definite as well.

Next, we will use the TP preserving property of superchannels to show that $C$ is not only positive semi-definite, but also has the desired form from Eq.~\eqref{eq:corr_form}. The TP condition for the output channel is equivalent to 
\begin{align}
		\frac{\delta_{kl}}{d} &=\sum_i \bra{ik} J(\E)\circ C\ket{il}
		=\sum_{i} 	C_{ik,il} J(\E)_{ik,il}.
\end{align}
As the above equality must hold for all channels $\E$, we must have that $C_{ik,il}$ does not depend on index $i$, and that the diagonal elements $C_{ik,ik}$ must be equal to one. This can be proven more explicitly by contradiction. Assume that for some $j,k$ we have $C_{jk,jk} \neq 1$, and consider a quantum channel $\E^{(0)}$ with Jamio{\l}kowski operator
\begin{equation}
    dJ(\E^{(0)}) = \ketbra{j}{j} \otimes \iden.
\end{equation}
We then see that
\begin{align}
		\sum_i \bra{ik} J(\E^{(0)})\circ C\ket{ik}
		= C_{jk,jk} J(\E^{(0)})_{jk,jk} \neq \frac{1}{d},
\end{align}
which contradicts our assumption of the TP preserving property. Similarly, the equality of the off-diagonal elements $C_{ik,il}$ for all $i$ can also be proved by contradiction. Assume that for some $i_0, i_1$ and $k \neq l$, we have $C_{i_0 k,i_0 l}  \neq C_{i_1 k,i_1 l}$. Now, consider a channel $\E^{(1)}$ with the Jamio{\l}kowski matrix given by
	\begin{align}
	    dJ(\E^{(1)})  =& \iden +(\ketbra{i_0}{i_0}-\ketbra{i_1}{i_1})\otimes(\ketbra{k}{l}+\ketbra{l}{k}).
	\end{align}
	To verify that the above matrix is a Jamio\l{}kowski matrix of a CPTP map is straightforward. We then have
	\begin{align}
	\sum_i \bra{ik} J(\E^{(1)}) \circ C\ket{il}	&=  \frac{C_{i_0 k, i_0 l} - C_{i_1 k, i_1 l}}{d} \neq 0,
	\end{align}
	and so $\Xi_C$ does not preserve the trace-preserving property, meaning that the assumption was wrong. We conclude that $\Xi_C$ is a superchannel if and only if the matrix $C$ has the form displayed in  Eq.~\eqref{eq:corr_form}. 
	
    Finally, we turn to proving that preserving the diagonal of the Jamio\l{}kowski state is equivalent to being a Schur-product superchannel. In order to keep the diagonal elements of a matrix $J(\E)$ unchanged by $\Xi$, the diagonal blocks of ${\bf J}_{\Xi}$ should be in the following form
	\begin{equation}\label{eqn:super-blocks}
		 {\bf J_\Xi}^{(i k)} := (\bra{i k } \otimes \iden) \  {\bf J}_{\Xi} \ (\ket{i k } \otimes \iden)  = \frac{1}{d^2}\ketbra{i k }{i k }.
	\end{equation}	 
	To see this, we will use Eq.~\eqref{eqn:super-action} and note
	\begin{align}
	\bra{i k } J(\Xi[\E]) \ket{ i k } &= d^2
	    \bra{i k }
	    \trr{2}{ \ {\bf J}_{\Xi} \  (\iden \otimes J(\E)^\top)}
	    \ket{i k } \nonumber\\
	    &= d^2\tr{ {\bf J_\Xi}^{(i k )}  J(\E)^\top}.
	\end{align}
	Our assumption was that diagonal elements of the Jamio\l{}kowski matrix must remain unchanged under the action of a superchannel $\Xi$, which is equivalent to the fact that for all channels $\E$ we have
	\begin{equation}
	   d^2 \tr{ {\bf J_\Xi}^{( i k)}  J(\E)^\top} = \bra{i k }J(\E) \ket{i k }.
	\end{equation}
	This gives us that the blocks ${\bf J_\Xi}^{( i k )}$ must be in the form presented in Eq.~\eqref{eqn:super-blocks}.
	
    So far we have proven that the condition for invariant diagonal elements of the Jamio\l{}kowski matrix of a channel under the action of superchannel $\Xi$ gives us the full diagonal of the Jamio\l{}kowski matrix of the superchannel, i.e.,
	\begin{equation}
	  \bra{i k i' k'}{\bf J_{\Xi}} \ket{i k i' k'} = \frac{\delta_{i i'} \delta_{k k'}}{d^2}.
	\end{equation}
	Now, we will use the fact, that the Jamio\l{}kowski matrix of a superchannel must be positive semi-definite~\cite{gour2019comparison}. For a non-negative matrix $A$ one has
	\begin{equation}
	    |A_{ij}|^2 \leq A_{ii} A_{jj}.
	\end{equation}
	Since in our case we have a lot of zeros on the diagonal, the elements of the Jamio\l{}kowski matrix of a superchannel $\Xi$ must satisfy the following inequalities 
	\begin{equation}
	 |\!\bra{i k i' k'}{\bf J_{\Xi}} \ket{j l j' l'}\!|^2
	 \leq  \frac{\delta_{i i'} \delta_{k k'} \delta_{j j'} \delta_{l l'}}{d^2}.
	\end{equation}
	The above implies that the Jamio\l{}kowski matrix of a superchannel $\Xi$ can be written as a sum defined in Eq.~\eqref{eqn:super=jamiolkowski}.
	
\end{proof}

%-------------------------------------------------------------------------------
% SEC. III.B
%-------------------------------------------------------------------------------

\subsection{Physical realisation}

The following proposition specifies the physical realisation of every Schur-product superchannel.

\begin{prop}[Physical realisation of $\Xi_C$]
	\label{prop:schur_super_rep}
	Every Schur-product superchannel can be written as a unitary pre- and post-processing with an ancillary system of dimension $d^2$ as follows:
	\begin{equation}
		\label{eq:diagram}
		\begin{quantikz}
			& \gate{\Xi_C[\E]} &\qw
		\end{quantikz}
		=
		\begin{quantikz}
			&[-0.1cm]	\gate[wires=2]{\U}&[-0.1cm] \gate{\E} &[-0.1cm] \gate[wires=2]{\V}&[-0.6cm]\qw \\[-0.5cm]
			\ketbra{0}{0}	&  	 & \qw 	&	& [-0.6cm]	\trash{\text{\emph{discard}}}
		\end{quantikz}\!\!\!,
	\end{equation}
	where
	\begin{subequations}
	\begin{align}
		\U(\cdot)&=U(\cdot)U^\dagger,\quad U=\sum_{i=1}^d \ketbra{i}{i}\otimes U_i, \\
		\V(\cdot)&=V(\cdot)V^\dagger,\quad V=\sum_{i=1}^d \ketbra{i}{i}\otimes V_i,
	\end{align}
	\end{subequations}
	with $\{U_i\},\{V_i\}$ being arbitrary unitaries of size $d^2$. The relation between $C$ and these unitaries is given by
	\begin{equation}
		\label{eq:C_superchanel}
		C_{ik,jl} = \bra{0} U_l^\dagger V_j^\dagger V_i U_k \ket{0}.
	\end{equation}
\end{prop}

\begin{proof}
	First, we assume the $\Xi_C[\E]$ has the form from Eq.~\eqref{eq:diagram} and calculate
	\begin{align}
	\!\!\!\Xi_C[\E](\ketbra{k}{l})&\!=\! \trr{2}{ \V\circ(\E \otimes \I)\circ\U [\ketbra{k}{l} \otimes \ketbra{0}{0}]}
	\nonumber\\
	&\!=\!
	\sum_{i j m n} \ketbra{i}{i} \E 
	(\ket{m}\!\braket{m}{k}\braket{l}{n}\!\bra{n} )
	\ketbra{j}{j}
	\nonumber\\
	&\qquad\qquad\times\tr {V_i U_m \ketbra{0}{0} U_{n}^\dagger V_j^\dagger}
	\nonumber\\
	&\!=\!
	\sum_{ij} \bra{i} \E 
	(\ketbra{k}{l})
	\ket{j}
	\bra{0} U_l^\dagger V_j^\dagger V_i U_k \ket{0}\ketbra{i}{j}
	\nonumber\\
	&\!=\!
	d\sum_{ij} J(\E)_{ik,jl} C_{ik,jl} \ketbra{i}{j}.\!\!
	\end{align}
	Since
	\begin{equation}
		J(\Xi_C[\E]) = \frac{1}{d}\sum_{k l}
		\Xi_C[\E](\ketbra{k}{l}) \otimes \ketbra{k}{l},
	\end{equation}
	we get
	\begin{equation}
		J(\Xi_C[\E]) = 	\sum_{ijkl} J_{ik,jl}(\E) C_{ik,jl} \ketbra{ik}{jl} =  J(\E)\circ C.
	\end{equation}
	Moreover, by direct inspection one can see that $C$ defined by Eq.~\eqref{eq:C_superchanel} is a correlation matrix with a block structure as in Eq.~\eqref{eq:corr_form}. Thus, by Proposition~\ref{prop:schur_super},
	the supermap $\Xi_C$ is a Schur-product superchannel.
	
	Next, we assume that $\Xi_C$ is a Schur-product superchannel. From Proposition~\ref{prop:schur_super}, we know that the matrix $C$ has equal diagonal sub-matrices, i.e.,
	\begin{equation}
		C_{ik,il} =  C_{1k,1l}.
	\end{equation}
	On the other hand, since $C$ is a correlation matrix it can be written as a Gram matrix:
	\begin{equation}
		C_{ik,jl} = \bra{\xi_{jl}} {\xi_{ik}}\rangle.
	\end{equation}
	We thus have
	\begin{equation}
	\bra{\xi_{il}} {\xi_{ik}}\rangle = \bra{\xi_{1l}} {\xi_{1k}}\rangle.
	\end{equation}
	Now, it is well known~\cite{jozsa2000distinguishability} that if two collections of vectors have the same Gram matrix, then the sets are related by a unitary transformation. Therefore,
	\begin{equation}
		\ket{\xi_{ik}} = V_{i} \ket{\xi_{1k}}
	\end{equation}
	for some unitary matrices $\{V_i\}$.  If we denote
	\begin{equation}
		\ket{\xi_{1 k}} = U_{k} \ket{0},
	\end{equation}
	with some unitary matrices $\{U_k\}$, then we obtain
	\begin{equation}
	C_{ik,jl}
	= 
	\bra{\xi_{1 l}} V_{j}^\dagger
	V_{i} \ket{\xi_{1 k}}
	=
	\bra{0} U^\dagger_{l}
	V_{j}^\dagger
	V_{i}
	U_{k}
	\ket{0}.
	\end{equation}
	Therefore, $\Xi_C$ can be written in the form from Eq.~\eqref{eq:diagram}.
\end{proof}

%-------------------------------------------------------------------------------
% SEC. III.C
%-------------------------------------------------------------------------------

\subsection{Comparison with dephasing pre- and post-processing}

By comparing Eq.~\eqref{eq:diagram} with Eq.~\eqref{eq:diagram_channel} we see that a dephasing superchannel acts as a generalisation of pre- and post-processing with dephasing channels that employs memory. More precisely, a dephasing channel simply correlates the system with an environment according to Eq.~\eqref{eq:diagram_channel}, and then the environment is discarded. But if instead it is kept intact, re-used again after the action of a channel $\E$ and only then discarded, the system would undergo evolution described by $\Xi_C[\E]$. The crucial question then is: how much more general transformations we can obtain due to these memory effects. In other words, we want to ask how much larger is the space of dephasing superchannels as compared to superchannels formed from pre- and post-processing by dephasing channels.

It is a straightforward calculation to show that pre- and post-processing by dephasing channels $\D_{C^{(1)}}$ and $\D_{C^{(2)}}$ has the following effect on the Jamio{\l}kowski state of a general channel $\E$:
\begin{equation}
    J(\D_{C^{(2)}}\circ\E\circ\D_{C^{(1)}})= J(\E)\circ (C^{(2)}\otimes C^{(1)}).
\end{equation}
Note that the symbol $\circ$ on the left hand side of the above equality denotes concatenation of channels, while on the right hand side it denotes the Schur product. We see that while a general dephasing superchannel is described by a correlation matrix $C$ from Eq.~\eqref{eq:corr_form}, the correlation matrices that we can obtain by pre- and post-processing without employing a memory are of the product form. In order to address our question we thus need to understand 
correlations carried by the bi-partite quantum
state associated with the matrix $C$ defined in Eq.~\eqref{eq:corr_form}.

Let us first note that if this state
is classically correlated, 
so the correlation matrix $C$ can be written in the form
\begin{equation}
    C=\sum_i p_i C^{(2,i)}\otimes C^{(1,i)},
\end{equation}
for some probability distribution $\v{p}$, then dephasing superchannels do not generate much more general transformations than dephasing pre- and post-processing without memory. It is because in this case they only correspond to probabilistic mixtures of various dephasing pre- and post-processings, i.e., they can be simulated by a classical coin toss followed by a dephasing pre- and post-processing dependent on the result
of the coin toss. Thus, dephasing superchannels can induce truly more general transformations only when the correlation matrix $C$ corresponds to an entangled state.

Interestingly, in the simplest case of a qubit system this is not the case and the correlation matrix $C$ can only be classically correlated. To see this, note that a general correlation matrix $C$ of the form from Eq.~\eqref{eq:corr_form} is given by
\begin{equation}
C=\begin{pmatrix}  
C_0&C_1\\
C_1^\dagger&C_0
\end{pmatrix}
\end{equation}
where $C_0$ and $C_1$ are $2\times 2$ matrices. Now, a unitary $\Pi$ given in the same block form by
\begin{equation}
    \Pi=\begin{pmatrix}
    0&\iden\\
    \iden&0
    \end{pmatrix}
\end{equation}
transforms $C$ into
\begin{equation}
C':=\Pi C\Pi^\dagger=\begin{pmatrix}  
C_0&C_1^\dagger\\
C_1&C_0
\end{pmatrix}.
\end{equation}
At the same time, the partial transpose of $C$ is given by
\begin{equation}
    C^{\top_2}=C'^{*},
\end{equation}
and so its spectrum is the complex conjugate of the spectrum of $C'$, which in turn is the same as the spectrum of $C$. Thus, $C$ has a positive partial transpose and by the Peres-Horodecki criterion~\cite{peres1996separability,horodecki2001separability} we know that $C$ is not entangled. We can thus conclude that every dephasing superchannel for qubit systems can be realised by a probabilistic mixture of dephasing pre- and post-processings.

However, already in dimension $d=3$, dephasing superchannels provide more general transformations. A simple example is given by a superchannel with the corresponding correlation matrix given by
\begin{equation}
C=\begin{pmatrix}
    \iden & A & B\\
    A^\dagger & \iden & 0\\
    B^\dagger & 0 & \iden
\end{pmatrix}
\end{equation}
where $\iden$ and $0$ denote identity and zero matrices of size 3, and
\begin{equation}
    A=\begin{pmatrix}
      0 & 0 & 0\\
      0 & 0 & 0\\
      1 & 0 & 0
    \end{pmatrix},\qquad
    B=\begin{pmatrix}
      1 & 0 & 0\\
      0 & 0 & 0\\
      0 & 0 & 0
    \end{pmatrix}.
\end{equation}
The above matrix $C$ is of the form \eqref{eq:corr_form} and it can be easily verified that its partial transpose has a negative eigenvalue $1-\sqrt{2}$. Thus, by Peres-Horodecki criterion, $C$ corresponds to an entangled state and so $\Xi_C$ cannot be obtained as a classical mixture of dephasing pre- and post-processing.

%-------------------------------------------------------------------------------
% SEC. III.D
%-------------------------------------------------------------------------------

\subsection{Action on Schur-product channels}

The action of a Schur-product superchannel $\Xi_{C}$ on a Schur-product channel $\D_{C'}$ is particularly simple and intuitive. We have
\begin{align}
	\Xi_{C}[\D_{C'}](\rho) & =  \sum_{ijkl}\rho_{kl}\ketbra{i}{i} \Xi_{C}[\D_{C'}](\ketbra{k}{l})\ketbra{j}{j}
	\nonumber\\
	&=d\sum_{ijkl}\rho_{kl} J(\Xi_{C}[\D_{C'}])_{ik,jl} \ketbra{i}{j}
	\nonumber\\
	&=d\sum_{ijkl}\rho_{kl} (J(\D_{C'})\circ C)_{ik,jl} \ketbra{i}{j}.
\end{align}
Now, since
\begin{equation}
	J(\D_{C'})=\frac{1}{d}\sum_{ij} C'_{ij} \ketbra{ii}{jj}, \label{Jdi}
\end{equation}
we get
\begin{align}
	J(\D_{C'})\circ C&=\frac{1}{d} \sum_{ij} C'_{ij} C_{ii,jj}\ketbra{ii}{jj}
	\nonumber\\
	&=\frac{1}{d}\sum_{ij} (C'\circ \tilde{C})_{ij}\ketbra{ii}{jj},\label{Jdi2}
\end{align}
where
\begin{equation}
	\label{eq:C_tilde}
	\tilde{C}_{ij}:=C_{ii,jj}
\end{equation}
is a correlation matrix due to Eq.~\eqref{eq:C_superchanel}. We thus arrive at
\begin{align}
	\!\!\Xi_{C}[\D_{C'}](\rho) &=\sum_{ij}\rho_{ij} (C'\circ\tilde{C})_{ij} \ketbra{i}{j}=\rho\circ(C'\circ \tilde{C})\nonumber\!\\
	\!\!&=\D_{C'\circ\tilde{C}}(\rho)=\D_{\tilde{C}}(\D_{C'}(\rho)).\!
\end{align}

Therefore, the action of $\Xi_{C}$ on a Schur-product channel $\mathcal{D}_{C'}$ is equivalent to post- (or pre-, since the considered channels commute) processing by another Schur-product channel~$\D_{\tilde{C}}$. In this particular case no memory is needed and the dephasing superchannel acts simply as a dephasing channel. Thus, $\Xi_{C}$ maps a dephasing channel $\mathcal{D}_{C'}$ to a more dephasing channel, i.e., the damping of coherence between states $i$ and $j$ originally described by $|C'_{ij}|$  becomes $|C'_{ij}\tilde{C}_{ij}|\leq|C'_{ij}|$. Moreover, all quantities that satisfy the data-processing inequality \cite{nielsen2010quantum}
are monotones, e.g., the capacity of a dephasing channel cannot increase under dephasing superchannel. 

%-------------------------------------------------------------------------------
% SEC. IV - DEPHASING SUPERCHANNELS AND COHERENCE OF CHANNELS
%-------------------------------------------------------------------------------

\section{Dephasing superchannels and coherence of channels}
\label{sec:coherence}

In this section we will discuss the interplay between dephasing superchannels and coherence resources of quantum channels they act upon. We will first present a short proof that the ability of a channel to create coherence deteriorates under the action of a dephasing superchannel. Then, we will explain that the power of dephasing superchannels to affect a quantum channel $\E$ (measured by the size of the orbit of channels that $\E$ can be sent to by dephasing superchannels) is bounded by the coherence content of $\E$. Finally, we will explain how coherence of a quantum channel $\E$ can be used as a resource to distinguish between various dephasing superchannels.

%-------------------------------------------------------------------------------
% SEC. IV.A
%-------------------------------------------------------------------------------

\subsection{Monotonicity of coherence generating power}

In Sec.~\ref{sec:schur_channels} we have mentioned that Schur-product channels cannot increase any meaningful measure of state's coherence $\C$~\cite{aberg2006superposition,baumgratz2014quantifying}, such as the $l_1$-norm of coherence or relative entropy of coherence. Here, we will prove an equivalent result for Schur-product superchannels. Namely, we will show that they cannot increase the cohering power $\C_g$~\cite{mani2015cohering}, which is a measure quantifying the ability of a quantum channel $\E$ to create coherence:
\begin{equation}
	\C_g(\E):=\max_k \{\C(\E(\ketbra{k}{k}))\},
\end{equation}
where $\C$ is any coherence measure satisfying the basic axioms~\cite{baumgratz2014quantifying}.

In order to achieve this, we will look at the action of a processed channel $\Xi_C[\E]$ on the distinguished diagonal states $\ketbra{k}{k}$. Employing the representation from Proposition~\ref{prop:schur_super_rep}, we have
\begin{align}
	\Xi_C[\E](\ketbra{k}{k})&=\sum_{ij} \matrixel{i}{\E(\ketbra{k}{k})}{j}\bra{0}U_k^\dagger V_j^\dagger V_i U_k\ket{0} \ketbra{i}{j}
	\nonumber\\
	&=\sum_{ij} \matrixel{i}{\E(\ketbra{k}{k})}{j} \tilde{C}_{ij} \ketbra{i}{j}
	\nonumber\\
	&=\E(\ketbra{k}{k})\circ \tilde{C}=\D_{\tilde{C}}(\E(\ketbra{k}{k})),
\end{align}
where we have introduced a correlation matrix $\tilde{C}$.

It is now straightforward to show that resource generating power is a monotone under $\Xi_C$:
\begin{align}
	\C_g(\Xi_C[\E]):=&\max_k \{\C(\Xi_C[\E](\ketbra{k}{k}))\}
	\nonumber\\
	=&\max_k \{\C(\D_{\tilde{C}}(\E(\ketbra{k}{k})))\}
	\nonumber\\
	\leq&\max_k \{\C(\E(\ketbra{k}{k}))\}=\C_g(\E),
\end{align}
where the inequality comes from the fact that Schur-product channels are incoherent operations and thus cannot increase any meaningful measure of coherence.

%-------------------------------------------------------------------------------
% SEC. IV.B
%-------------------------------------------------------------------------------

\subsection{Power of dephasing superchannels}

We now proceed to investigating how strongly can a quantum gate, described by a channel $\E$, be affected by a dephasing noise described by a family of dephasing superchannels. Intuitively, one expects that dephasings can affect a more coherent gate more strongly. For example, in the extreme case of a classical channel,
\begin{equation}
    \E_T(\cdot):=\sum_{ij} T_{ij} \matrixel{j}{(\cdot)}{j} \ketbra{i}{i}\label{classicalC}
\end{equation}
with $T$ being a stochastic matrix, one has \mbox{$\Xi_C[\E]=\E$} independently of $C$. Thus, classical channels are unaffacted by a dephasing noise. On the other hand, arguably the most coherent qubit channel given by
\begin{equation}
    \E(\cdot)=H(\cdot)H^\dagger,\qquad H=\frac{1}{\sqrt{2}}\begin{pmatrix}1&1\\1&-1\end{pmatrix}
\end{equation}
can be sent to a perfectly distinguishable channel $\E'$ by $\Xi_C$ given by
\begin{equation}
    C=\begin{pmatrix}
        \phantom{-}I&-I\phantom{-}\\
        -I&\phantom{-}I\phantom{-}
    \end{pmatrix},\qquad
    I:=\begin{pmatrix}
        1&1\\
        1&1
    \end{pmatrix}.
\end{equation}
Here, perfectly distinguishable means that there exists an input state $\ketbra{0}{0}$ that is mapped by $\E$ and $\E'$ to two orthogonal quantum states,
$\ketbra{+}{+}$ and $\ketbra{-}{-}$, which can be distinguished in an experiment with probability equal to 1.

One way to quantify the effect that a dephasing noise can have on a channel $\E$
is to find ``how far'' a quantum channel $\E$ can be send via a dephasing superchannel. In other words, we wish to
evaluate the supremeum,
\begin{equation}
    \sup_C D(\E, \Xi_C[\E]),
\end{equation}
with $D$ being some distance measure on the set of channels, e.g., the diamond norm distance. Here, we will investigate a more coarse-grained notion: we will look for the maximal number of distinguishable channels that can be obtained from a given channel $\E$ via dephasing superchannels. In other words, we want to find the size of the error space, i.e., the size of the image of $\E$ under the action of all possible dephasing superchannels. More formally, our aim is to upper-bound the number $M(\E,\epsilon)$, which is the maximal number of channels 
\begin{equation}
    \E_m:=\Xi_{C_m}[\E]
\end{equation} 
that can be obtained from a given channel $\E$ via dephasing superchannels $\Xi_{C_m}$ and that are distinguishable with average probability $1-\epsilon$. Recall that a general scheme for distinguishing between $M$ channels acting on $d_A$-dimensional states is composed of an input state $\rho^{AB}$ on an extended space $\H_{d_A}\otimes \H_{d_B}$, together with a decoding measurement described by POVM elements $\{E^{AB}_m\}_{m=1}^M$. The average probability of distinguishing between the channels is then given by
\begin{equation}
    1-\epsilon=\frac{1}{M}\sum_{m=1}^{M}\textrm{Tr}\left( E^{AB}_{m}(\E^A_m\otimes\I^B)(\rho^{AB})\right).
\end{equation}
Optimising the above over all input states $\rho^{AB}$ and measurements  $\{E_m^{AB}\}_{m=1}^M$ yields the optimal distinguishing probability. Note that to achieve optimal probability, it is enough to choose $d_B=d_A$; however, for technical reasons, we choose it to be $d_B=Md_A$.

We start by introducing the notation for a completely dephasing channel
\begin{equation}
    \Delta(\cdot)=\sum_i \matrixel{i}{(\cdot)}{i} \ketbra{i}{i}
\end{equation}
and for classical (completely dephased) version of a channel $\E$:
\begin{equation}
    \E_\Delta:=\Delta\circ\E\circ\Delta.
\end{equation}
We note that all dephasing superchannels $\Xi_C$ satisfy
\begin{subequations}
\begin{align}
    \Xi_{C}\left(\E_\Delta\right)&=\E_\Delta\label{eq:cond1},\\
    \Delta\circ\Xi_{C}\left(\E\right)\circ\Delta&=\E_\Delta,\label{eq:cond2}
\end{align}
\end{subequations}
for all channels $\E$. We also introduce the following two classical-quantum states:
\begin{subequations}
\begin{align}
    \!\!\!\tau^{RAB}&:=\frac{1}{M} \!\sum_{m=1}^M \ketbra{m}{m}^R\otimes (\E^A_m\otimes \I^B)(\rho^{AB}),\\
    \!\!\!\zeta^{RAB}&:=\frac{1}{M} \!\sum_{m=1}^M \ketbra{m}{m}^R\otimes (\E^A_\Delta\otimes \I^B)(\rho^{AB}),\!
\end{align}
\end{subequations}
for some input state $\rho^{AB}$, and recall the notion of hypothesis testing relative entropy~\cite{wang10,buscemi2010quantum,brandao2011one}:
\begin{align}
\!\!\!D_H^{\epsilon}(\rho\|\sigma) := - \log \inf \big\{ \tr{Q\sigma}\ \big|\ \!&0 \leq Q \leq \iden,\nonumber\\
\!&\tr{Q\rho} \geq 1-\epsilon \big\} \,.
\label{eq:hypothesis_rel_ent}
\end{align}

Now, let us assume that there exists a choice of superchannels $\{\Xi_{C_m}\}_{m=1}^M$ such that $M$ resulting channels $\E_m=\Xi_{C_m}(\E)$ are distinguishable with average probability larger than \mbox{$1-\epsilon$}. This means that there exist an input state $\rho^{AB}$ and a POVM measurement $\{E^{AB}_m\}_{m=1}^M$ such that
\begin{equation}
    \frac{1}{M}\sum_{m=1}^{M}\textrm{Tr}\left( E^{AB}_{m}(\E^A_m\otimes\I^B)(\rho^{AB})\right)\geq 1-\epsilon.
\end{equation}
We can thus introduce an operator $Q$:
\begin{equation}
    Q:=\sum_{m=1}^M \ketbra{m}{m}^R\otimes E_m^{AB},
\end{equation}
which satisfies $0\leq Q\leq \iden$ and
\begin{equation}
    \tr{Q\tau^{RAB}}\geq 1-\epsilon.
\end{equation}
As a result we have the following bound:
\begin{equation}
    D^\epsilon_H(\tau^{RAB}\|\zeta^{RAB})\geq -\log (Q\zeta^{RAB}) = \log M.
\end{equation}
To get a state-independent bound we optimise over all input states $\rho^{AB}$ to arrive at
\begin{equation}
    \log [M(\E,\epsilon)] \leq \sup_{\rho^{AB}} D^\epsilon_H(\tau^{RAB}\|\zeta^{RAB}).
\end{equation}

The next step is to bring $\tau^{RAB}$ and $\zeta^{RAB}$ to a more useful form. For that we need to introduce an auxiliary state 
\begin{equation}
 \sigma^{RAB}:=\frac{1}{M}\sum_{m=1}^{M}\ketbra{m}{m}^R\otimes\rho^{AB}
\end{equation}
and a superoporator
\begin{equation}
    \M_m(\cdot)=\ketbra{m}{m}(\cdot)\ketbra{m}{m}.
\end{equation}
With a slight abuse of notation we will also use $\M_m$ to denote a supermap that is acting as a post-processing via $\M_m$, i.e., $\M_m[\E]:=\M_m\circ \E$. We now have the following
\begin{align}
    \tau^{RAB}&= \sum_{m=1}^M (\M_m^R\otimes \E_m^A\otimes \I^B)(\sigma^{RAB})\nonumber\\
    &=\left(\sum_{m=1}^M \M^R_m\otimes\Xi^A_m\otimes\I^B\right)\nonumber\\
    &\qquad\qquad[\I^R\otimes\E^A\otimes\I^B](\sigma^{RAB}),
\end{align}
where the first parentheses contains a superchannel that acts on the channel in the second parentheses. Similarly, we also have
\begin{align}
    \zeta^{RAB}&=\sum_{m=1}^M (\M_m^R\otimes \E_\Delta^A\otimes \I^B)(\sigma^{RAB})\nonumber\\
    &=\sum_{m=1}^M (\M_m^R\otimes \Xi_m[\E_\Delta^A]\otimes \I^B)(\sigma^{RAB})\nonumber\\
    &=\left(\sum_{m=1}^M \M^R_m\otimes\Xi^A_m\otimes\I^B\right)\nonumber\\
    &\qquad\qquad[\I^R\otimes\E_\Delta^A\otimes\I^B](\sigma^{RAB}),
\end{align}

With the following short-hand notation:
\begin{subequations}
\begin{align}
    \Theta^{RA}&:=\sum_{m=1}^M \M^R_m\otimes\Xi^A_m,\\
    \E^{RA}&:=\I^R\otimes \E^A,\\
    \E_\Delta^{RA}&:=\I^R\otimes \E_\Delta^A,
\end{align}
\end{subequations}
we then have 
\begin{align}
    \!\!\!\log[M(\E,\epsilon)] &\leq \sup_{\rho^{AB}} D_H^\epsilon \left(\Theta^{RA}[\E^{RA}]\otimes \I^B (\sigma^{RAB})\right.\nonumber\\
    &\qquad\qquad\qquad\left. \|\Theta^{RA}[\E_\Delta^{RA}]\otimes \I^B (\sigma^{RAB})\right)\nonumber\\
    &\leq \sup_{\sigma^{RAB}} D_H^\epsilon \left(\Theta^{RA}[\E^{RA}]\otimes \I^B (\sigma^{RAB})\right.\nonumber\\
    &\qquad\qquad\qquad\left. \|\Theta^{RA}[\E_\Delta^{RA}]\otimes \I^B (\sigma^{RAB})\right)\nonumber\\
   & =: \mathcal{C}_{D_H^\epsilon}(\Theta^{RA}[\E^{RA}]\|\Theta^{RA}[\E_\Delta^{RA}]),
\end{align}
where $\mathcal{C}_{D}$ is the channel divergence introduced in Ref.~\cite{gour2019comparison} for any state divergence measure $D$:
\begin{align}
\mathcal{C}(\E_1\|\E_2)&=\sup_{\rho^{AB}} D((\E^A_{1}\otimes\mathcal{I}^B)(\rho^{AB})\nonumber\\
&\qquad\qquad\qquad \|(\E^A_{2}\otimes\mathcal{I}^B)(\rho^{AB})). \label{eq:classpross}
\end{align}
Importantly, channel divergences satisfy data-processing inequality, and so
\begin{align}
    \log[M(\E,\epsilon)] &\leq \mathcal{C}_{D_H^\epsilon}(\E^{RA}\|\E_\Delta^{RA})=\mathcal{C}_{D_H^\epsilon}(\E\|\E_\Delta).
\end{align}

We thus conclude that the effect that dephasing noises can have on a given quantum gate $\E$, as quantified by $M(\E,\epsilon)$, is upper-bounded by channel coherence measure related to hypothesis testing relative entropy:
\begin{equation}
    M(\E,\epsilon) \leq 2^{\mathcal{C}_{D_H^\epsilon}(\E\|\E_\Delta)}.
\end{equation}

%-------------------------------------------------------------------------------
% SEC. IV.C
%-------------------------------------------------------------------------------

\subsection{Distinguishing dephasing superchannels}

To complement the discussion from the previous section, here we explain how the sensitivity to dephasing noises of a coherent channel can be considered as a resource for distinguishing dephasing superchannels. As already observed, classical channels are invariant under dephasing superchannels, and so they cannot be used to distinguish between any two dephasing noises. On the other hand, a coherent channel is transformed non-trivially, so the resulting channel should carry some information about the parameters of the corresponding dephasing superchannel, and hence should be more helpful for noise metrology. 

Here we will show how coherence of a channel $\mathcal{E}$ quantified by generalised robustness of coherence upper bounds the number of dephasing superchannels that can be distinguished using $\E$. More precisely, given a channel $\E$ and a set of dephasing superchannels $\{\Xi_{C_i}\}$, a general strategy to distinguish between the elements of this set is to apply the processed channel $\Xi_{C_i}[\E]$ to half of a bipartite (possibly entangled) state $\rho^{AB}$ and to perform a measurement on the resulting state. The optimal success probability of distinguishing between $M$ uniformly sampled dephasing superchannels is then given by
\begin{align}
 \!\! &\!\!p_{\rm succ}(\{\Xi_{C_i}\},\E)\nonumber\\
 \!\! &\!\! :=\!\!\!\!\!\!\max_{\rho^{AB},\{E^{AB}_i\}}\! \frac{1}{M}\!\sum_{i=1}^M \tr{E_i^{AB}(\Xi_{C_i}[\E^A]\otimes\I^B) (\rho^{AB})},\!
 \label{eq:p_succ}
\end{align}
where $\rho^{AB}$ is maximised over all bipartite input states and $\{E_i^{AB}\}$ over all joint decoding POVM elements. In what follows our aim will be to upper bound the maximum number $M(\E,\epsilon)$ of dephasing superchannels distinguishable using $\E$ with probability $1-\epsilon$. 

To achieve the above goal we will use the concept of robustness of coherence, originally introduced as a measure of coherence for quantum states in Ref.~\cite{piani2016robustness}, and recently generalised to quantify the coherence of channels in Ref.~\cite{TR}. First, let us slightly abuse the notation and denote by $\mathcal{T}_c$ the set of of classical channels $\E_T$ defined in Eq.~\eqref{classicalC}, i.e., with Jamio{\l}kowski states $J\left(\E_T\right)$ given by
\begin{equation}
    J(\E_T)=\frac{1}{d}\sum_{i,j}T_{ij}\ketbra{i}{i}\otimes\ketbra{j}{j},\label{eq:classicalJ}
\end{equation}
which are incoherent in the distinguished basis. Note that the set $\mathcal{T}_c$ is convex and closed. Now, the generalized robustness of coherence of a channel $\E$ is defined as
\begin{equation}
    R\left(\E\right):=\min_{\F\in \T_q}\left\{ r\geq 0\mid\frac{\E+r \F}{1+r}\in\mathcal{T}_c\right\},
\end{equation}
where the minimum is taken over the set of all quantum channels~$\T_q$.
The generalized robustness $R(\E)$ quantifies the minimum amount of noise a channel $\E$ can withstand before becoming classical.

Next, for a given channel $\E$ let us denote by $\F_*$ the channel achieving the minimum in the definition of robustness, and by the $\tilde{\E}$ the resulting classical channel, i.e.,
\begin{equation}
    \tilde{\E}:=\frac{\E+R(\E)\F_*}{1+R(\E)}\in \T_c.
\end{equation}
Inverting the above to get the expression for $\E$, the following then holds for any bipartite state $\rho$ and any measurement $\{E_i\}$ (to simplify notation we drop the superscripts denoting subsystems):
\begin{align}
    &\sum_{i=1}^M \tr{E_i(\Xi_{C_i}[\E]\otimes\I) (\rho)} \nonumber\\
    &\quad = \sum_{i=1}^M \tr{E_i(\Xi_{C_i}[(1+R(\E))\tilde{\E}-R(\E)\F_*]\otimes\I) (\rho)}\nonumber\\
    &\quad \leq (1+R(\E))\sum_{i=1}^M \tr{E_i(\Xi_{C_i}[\tilde{\E}]\otimes\I) (\rho)}\nonumber\\
    &\quad = (1+R(\E))\sum_{i=1}^M \tr{E_i(\tilde{\E}\otimes\I) (\rho)}\nonumber\\
    &\quad = (1+R(\E)),
\end{align}
where we used the fact that a classical channel $\tilde{\E}$ is invariant under dephasing superchannels $\Xi_{C_i}$. Using the above together with Eq.~\eqref{eq:p_succ} and denoting the success probability by $(1-\epsilon)$, we arrive at the upper bound for the number $M(\E,\epsilon)$:
\begin{equation}
    \label{eq:bound1}
    M(\E,\epsilon) \leq \frac{1+R(\E)}{1-\epsilon}.
\end{equation}
Note that the right hand side of Eq.~\eqref{eq:bound1} is a function of the generalized robustness, and so it determines an efficiently computable upper bound for $M(\mathcal{E},\epsilon)$~\cite{TR}. Moreover, the obtained bound quantitatively demonstrates the intuitive claim that a coherence content of $\mathcal{E}$ is a necessary resource for distinguishing dephasing noises.

%-------------------------------------------------------------------------------
% SEC. V - CONCLUSIONS AND OUTLOOK
%-------------------------------------------------------------------------------

\section{Conclusions and outlook}
\label{sec:outlook}

In this work we introduced the most natural class of superchannels that model dephasing noises acting on quantum gates. We provided their mathematical representation and physical realisation analogous to those of dephasing channels, but also proved that they describe a wider class of noises. Furthermore, we applied our characterisation to determine several effects of dephasing noises on quantum channels, such as the decrease in coherence generating power or the maximum possible disturbance. Additionally, we demonstrated that our formalism allows one to exploit the sensitivity of coherent gates to dephasing noises as a resource in the field of noise metrology. 

The results presented here should form a timely contribution to the development of quantum technologies, where the control of noise remains a significant challenge. Moreover, our formalism could be of interest for current research lines on superchannels with memory in time, or parallel correlations~\cite{ Yokojima2021}. The simplicity of the model studied here could be helpful to develop a tractable case study in the above mentioned research, with concrete implementations in diamond nitrogen-vacancy centres~\cite{Chen2018} and efficient quantum error correction codes~\cite{Layden2020}. Among the quantum technologies that we expect to take advantage of the current contribution one could mention quantum heat engines~\cite{dag2019temperature} and the quantum internet~\cite{Wehner18}. As any quantum communication network, the quantum internet requires calibration of the noise between nodes, but on the early stages of development it might lack local memories and access to quantum error correction. Precisely, in the mentioned early stages of the network, our formalism would provide a tool to estimate and model the errors in coherent operations between nodes. 

Another research area that could benefit from the results presented in this work is the theory of channel resources~\cite{Gour2019,Liu2020,Li2020,Carlo2020,Winter2019}. 
Our results imply that dephasing superchannels are good candidates for free operations in resource theories of coherence generating power~\cite{Li2020,Liu2020}. In consequence, they could be employed to compute lower bounds on channel distillation rates in these theories~\cite{Liu2020}. Finally, we would like to emphasize that our work lays the foundation for similar investigations to characterize different classes of gate noises. Some natural analyses of noise along parallel lines should include amplitude-damping and leakage or random unitary errors (especially the uniform depolarisation). The extensions of the latter to the context of quantum gates could bring significant progress in the theory of quantum control of noise.

%-------------------------------------------------------------------------------
% ACKNOWLEDGEMENTS
%-------------------------------------------------------------------------------

\subsection*{Acknowledgements} 

We would like to thank Dariusz Chru\'{s}ci\'{n}ski for useful comments on the manuscript. Z.P., K.K. and R.S. acknowledge financial support by the Foundation for Polish Science through TEAM-NET project (contract no. POIR.04.04.00-00-17C1/18-00). P.H. acknowledges support by the Foundation for Polish Science (IRAP project, ICTQT, contract no. 2018/MAB/5), co-financed by EU within Smart Growth Operational Programme. K.\.Z. is supported by National Science Center in Poland
under the Maestro grant number DEC-2015/18/A/ST2/00274.

\bibliography{Bib_channels}

\end{document}